\documentclass[10pt,twocolumn]{IEEEtran}

\usepackage{amsthm, amssymb}

\def\snr{\textrm{SNR}}

\def\dint{{\rm \ d}}

\newtheoremstyle{slplain}
  {3pt}
  {3pt}
  {\slshape}
  {}
  {\bfseries}
  {.}%
  { }
  {}

\theoremstyle{slplain}

\newtheorem{cor}{Corollary}
\newtheorem{lem}{Lemma}
\newtheorem{pro}{Proposition}

\usepackage{cite}
\usepackage{amsfonts}
\usepackage{subfigure}
\usepackage{multirow}

\ifCLASSINFOpdf
  \usepackage[pdftex]{graphicx}
  \graphicspath{{../pdf/}{../jpeg/}}
  \DeclareGraphicsExtensions{.pdf,.jpeg,.png}
\else
  \usepackage{float}
  \usepackage[dvips]{graphicx}
  \graphicspath{{../eps/}}
  \DeclareGraphicsExtensions{.eps}
\fi

\usepackage[cmex10]{amsmath}
\usepackage{algorithm}
\usepackage{algorithmic}
\usepackage{array}
\usepackage{url}

\begin{document}

\title{{ \small \copyright 20xx IEEE. Personal use of this material is permitted. Permission from IEEE must be obtained for all other uses, in any current or future media, including reprinting/republishing this material for advertising or promotional purposes, creating new collective works, for resale or redistribution to servers or lists, or reuse of any copyrighted component of this work in other works.} \\ {\LARGE An Optimal Stopping Approach to Cell Selection in 5G Networks}}

\author{
\IEEEauthorblockA{Ning Wei, Xingqin Lin, Guangrong Yue, and Zhongpei Zhang} 
\thanks{ 

Ning Wei, Guangrong Yue, and Zhongpei Zhang are with the National Key Laboratory of Science and Technology on Communications,
University of Electronic Science and Technology of China, Sichuan, China. (Email: wn, yuegr, zhangzp@uestc.edu.cn.) 

Xingqin Lin is with Ericsson Research, Santa Clara, CA, USA. (Email: xingqin.lin@ericsson.com.) 

This work was supported by National Natural Science Foundation of China (Grant No.'s 61871070, 61831004).

An earlier version of this work has been uploaded to arXiv: https://arxiv.org/abs/1811.07452.


}
}

\maketitle

\begin{abstract}
Initial cell search and selection is one of the first few essential steps that a mobile device must perform to access a mobile network. The distinct features of 5G bring new challenges to the design of initial cell search and selection. In this paper, we propose a load-aware initial cell search and selection scheme for 5G networks. The proposed scheme augments the existing pure received power based scheme by incorporating a new load factor broadcast as part of system information. We then formulate a throughput optimization problem using the optimal stopping theory. We characterize the throughput optimal stopping strategy and the attained maximum throughput. The results show that the proposed cell search and selection scheme is throughput optimal with a carefully optimized connection threshold. 
\end{abstract}



\section{Introduction}

In June 2018, the third-generation partnership project (3GPP) approved the technical specifications for the standalone version of the 5th generation (5G) wireless access technology, known as new radio (NR) \cite{lin20185g}. Meanwhile, 3GPP continues evolving long-term evolution (LTE) to meet 5G requirements. In this paper, we study initial cell search and selection in 5G networks to achieve better load balancing and optimized throughput performance.

In cellular networks, a user equipment (UE) needs to perform essential cell search and selection prior to data communication. Cell search usually involves two steps: 1) search and acquire synchronization to a cell, and 2) decode system information that contains essential information for accessing the cell \cite{dahlman20134g}. In initial cell search, the UE shall search for the strongest cell  \cite{3gpp36304}. Once the UE finds an appropriate cell, it can select the cell and proceed to establish a connection with the cell by performing random access \cite{lin2016random}.

Although the received power based initial cell selection works well in single-tier homogeneous networks, it may result in unbalanced loads across the cells in multi-tier heterogeneous networks where the base stations (BSs) of small cells have much lower transmit powers than the BSs of macro cells \cite{andrews2014overview, bhushan2014network}. Further, millimeter wave (mmWave)  is a distinct feature of 5G networks \cite{lin20185g, andrews2014will}. MmWave communication utilizes the large chunks of spectrum resources in the mmWave bands (30-300 GHz) to achieve multi-Gbps data rates \cite{rappaport2013millimeter}.  The load imbalance issue may become more serious in 5G networks where small cells would likely operate on the mmWave spectrum. If there were few connections in the mmWave small cells, the abundant radio resources brought by mmWave spectrum would not be sufficiently utilized. 

Cell selection in 5G networks has drawn much interest in the past few years. We overview a few exemplary works related to our discussions herein. The work \cite{karimi2018centralized} proposed a centralized joint cell selection and scheduling policy for ultra-reliable low latency communication (URLLC) in the context of 5G NR networks. In \cite{elkourdi2018towards}, the authors also aimed to optimize latency in 5G heterogeneous networks by proposing a Bayesian cell selection algorithm that incorporates access nodes capabilities and UE traffic type. The cell selection in 5G heterogeneous networks was analyzed in \cite{gharam2018cell} using a non-cooperative game-theoretic framework with a mixed strategy Nash equilibrium method. The authors in \cite{sherin2018cell} also studied cell association in a two-tier 5G heterogeneous network by using an evolutionary game approach. In the context of LTE advanced heterogeneous networks, the work \cite{lai2018cell} addressed the cell selection and resource allocation optimization problem by dynamic programming. The work \cite{chang2018adaptive} considered 5G cooperative cellular communication and proposed to use adaptive heading prediction of moving path for user-centric cell selection. In \cite{ba2019load}, the authors proposed  a dynamic serving cell selection scheme based on channel state information and cell load for multi-connectivity in a 5G ultra dense network. In this paper, we adopt an optimal stopping approach to cell selection in 5G networks. To the best of our knowledge, optimal stopping theory has not previously been used to tackle cell selection in 5G networks.

Optimal stopping theory is about determining a time to take a given action based on causal observations to maximize an expected reward \cite{peskir2006optimal, ferguson2012optimal}. The application of optimal stopping theory to wireless communications and networking can be found in a diverse set of problems \cite{chaporkar2008optimal, zheng2009distributed, jiang2009optimal, zhang2010multi, cheng2011simple, poulakis2013channel, yang2014dynamic, wei2015optimal, li2015distributed, cacciapuoti2016optimal}. For example, the early work \cite{chaporkar2008optimal} applies optimal stopping theory to study joint probing and transmission strategies to maximize the system throughput in a broadcast fading channel. More recently, optimal stopping theory has been used to study emerging problems in the areas such as mmWave cellular systems \cite{wei2015optimal}, energy harvesting based wireless networks \cite{li2015distributed}, coexistence of heterogeneous networks in TV white space \cite{cacciapuoti2016optimal}, and cellular in unlicensed spectrum \cite{wei2017throughput}. 

In this paper, we propose a load-aware initial cell search and selection scheme to achieve a more balanced load distribution in 5G networks. In the proposed scheme, cell load information is included in the broadcast system information to facilitate initial cell search and selection. The proposed scheme augments the existing pure received power based initial cell search and selection scheme with a new load factor broadcast as part of the system information. This equips the networks with a powerful access control method that facilitates load balancing. 

We then study the theoretical performance of the proposed initial cell search and selection scheme. We formulate a throughput optimization problem. Using the optimal stopping theory \cite{peskir2006optimal, ferguson2012optimal}, we cast the problem as a maximal rate of return problem. We characterize the throughput optimal stopping strategy and the attained maximum throughput. The results show that the proposed initial cell search and selection scheme is throughput optimal with a carefully optimized connection threshold. The optimal connection threshold and maximum throughput can be found by solving a fixed point equation. The fixed point equation in general does not admit a closed-form solution. We further provide an alternative characterization of the optimal throughput and use it to develop an iterative algorithm to compute the solution to the fixed point equation. We also prove the convergence of the proposed iterative algorithm.

This rest of this paper is organized as follows. Section \ref{sec:proposed} describes the proposed initial cell search and selection scheme. Section \ref{sec:model} introduces the system model and mathematically formulates the problem. Section \ref{sec:optimal} analyzes the problem in detail using optimal stopping theory. Section \ref{sec:sim} provides simulation results to demonstrate the analytical results and obtain insights. Finally, Section \ref{sec:conclusion} presents our conclusions.

\section{Proposed Initial Cell Search and Selection}
\label{sec:proposed}

To achieve a more balanced load distribution in 5G networks, we propose that some cell load information is included in the broadcast system information to facilitate initial cell search and selection. Such load information may take various forms such as the number of connections in the cell, the available bandwidth, or the expected scheduling probability. To reduce access latency, the load information shall be broadcast in the first system information block. Such system information block usually  consists of a limited number of important bits that shall be read by all the UEs before accessing the cell.

For concreteness, we assume that each cell $i$ broadcasts the expected scheduling probability $\beta_i \in [0, 1]$ to control UE access. In other words, $\beta_i$ would be treated by UE in the initial cell search and selection as the expected scheduling probability if the UE camps on the cell. The proposed expected scheduling probability may be considered as ``soft'' access barring. Access barring has been used in cellular networks to reduce the access load in case of an overload situation \cite{3gpp36331}. Compared to the existing access barring that takes a boolean value (i.e., access allowed or denied), the proposed expected scheduling probability takes value in $[0, 1]$ and thus is a soft access barring scheme.

We are now in a position to describe the proposed initial cell search and selection scheme. When examining a cell $i$, UE first obtains synchronization and measures the corresponding received signal-to-noise ratio (SNR) from the BS. If the measured received SNR from BS $i$ (denoted as $\snr_i$) is below some threshold $\Gamma$, UE stops examining the current cell $i$ and starts to examine another cell. Here, the threshold $\Gamma$ is introduced to avoid selecting a cell with a too low SNR that may lead to poor link quality. If $\snr_i \geq \Gamma$, UE proceeds with decoding the first system information block to extract the cell load information, i.e., the expected scheduling probability $\beta_i$. Then UE computes the selection metric denoted as $R_i$ as follows:
\begin{align} 
R_i = \beta_i \log ( 1 + \snr_i  ) .
\label{eq:1}
\end{align}
If $R_i$ is greater than or equal to some connection threshold $\mu$, UE selects the cell $i$ and completes the cell search and selection process. Otherwise, UE repeats the process with another cell.

The proposed initial cell search and selection scheme is simple yet powerful. It augments the existing pure received power based scheme with a new load factor broadcast as part of the system information. It is up to the BS to decide and broadcast the value of the load factor, i.e., the expected scheduling probability $\beta_i$. For example, assuming that the number of active UEs served by BS $i$ is $M_i$, a UE in the initial cell search may expect its scheduling probability to be $1/(M_i + 1)$ if it selects BS $i$ as its serving cell. Therefore, the value of $\beta_i$ broadcast by BS $i$ may be chosen to be ${1}/(M_i + 1)$.

As an illustrative example, Figure \ref{fig:1} shows a realization of load distribution under different association schemes in a two-tier network. In the network,  a macro BS is located at the center $(0, 0)$ m, four micro BSs are respectively located at $(100, 100)$ m, $(-100, 100)$ m, $(-100, -100)$ m, and $(100, -100)$ m, and $100$ UEs are uniformly distributed. The macro BS operates in the $2$ GHz band with $20$ MHz carrier bandwidth, and the four micro BSs operate in the $39$ GHz band with $1$ GHz carrier bandwidth. The transmit powers of the macro BS and the micro BSs are respectively $46$ dBm and $23$ dBm. A total $30$ dB beamforming gain is assumed for each link between micro BSs and UEs. The noise power spectral density is $-174$ dBm/Hz. The UE noise figure is $9$ dB. The path loss of each link is equal to $20\log_{10} \left( \frac{4\pi}{\varrho} \right) + \alpha \cdot 10\log_{10} ( d )$ dB, where $\varrho$ is the wavelength, $\alpha=3.8$ is the path loss factor, and $d$ is the length of the radio link. Each link is also subject to a random log-normal shadowing with $7$ dB standard deviation.

Figure \ref{fig:1001} shows the load distribution under max-received-power association. The number of UEs selecting the macro BS is $76$, while the numbers of UEs selecting the four micro BSs are $5, 4, 6$, and $9$, respectively. Figure \ref{fig:1002} shows the load distribution under max-SNR association. Due to the high noise floor associated with the large mmWave bandwidth, the numbers of UEs selecting the four micro BSs become even fewer. Specifically, in Figure \ref{fig:1002}, the number of UEs selecting the macro BS is $97$, while the numbers of UEs selecting the four micro BSs are $1, 1, 0$, and $1$, respectively. Figure \ref{fig:1003} shows the load distribution under the proposed max-selection-metric association, with the selection metric defined in (\ref{eq:1}). In this simulation, the value of $\beta_i$ broadcast by BS $i$ is equal to ${1}/(M_i + 1)$. Initially, $M_i = 0$. UE performs cell selection one by one. If a UE selects BS $i$, $M_i$ is incremented by $1$, and $\beta_i$ is updated accordingly. For the purpose of illustration, the SNR threshold $\Gamma$ is set to be $-\infty$, and the connection threshold $\mu$ of each UE is chosen such that each UE selects the BS that yields the maximum value of the selection metric for the UE. In Figure \ref{fig:1003}, the number of UEs selecting the macro BS is $25$, while the numbers of UEs selecting the four micro BSs are $21, 15, 20$, and $19$, respectively. Clearly, compared to the load distribution in Figure \ref{fig:1001} or Figure \ref{fig:1002}, this is a much more balanced load distribution that sufficiently utilizes the abundant radio resources brought by the mmWave spectrum.

\begin{figure}[ht]
 \centering
 \subfigure[]{
  \includegraphics[width=8.5cm]{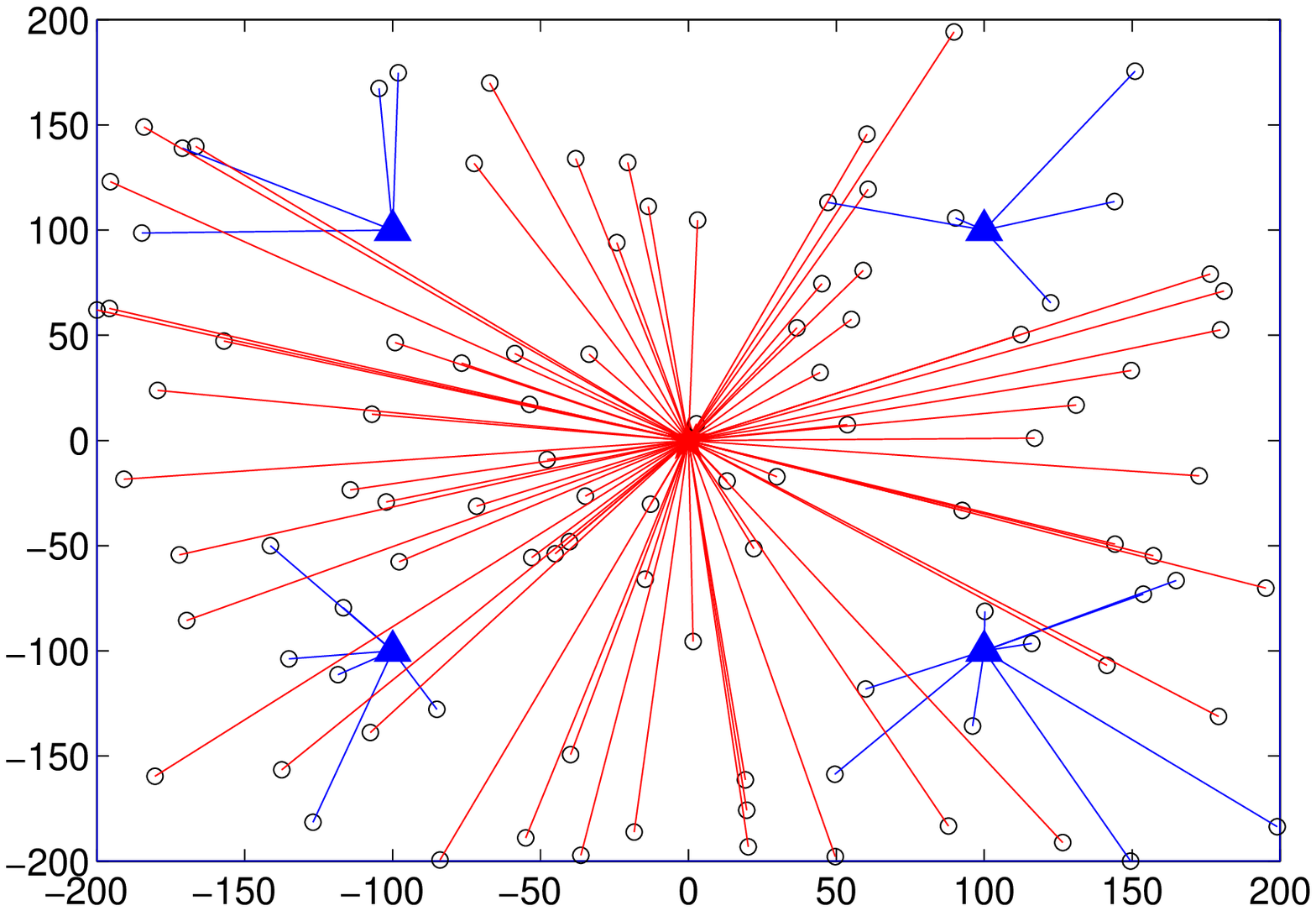}
   \label{fig:1001}
   }
 \subfigure[]{
  \includegraphics[width=8.5cm]{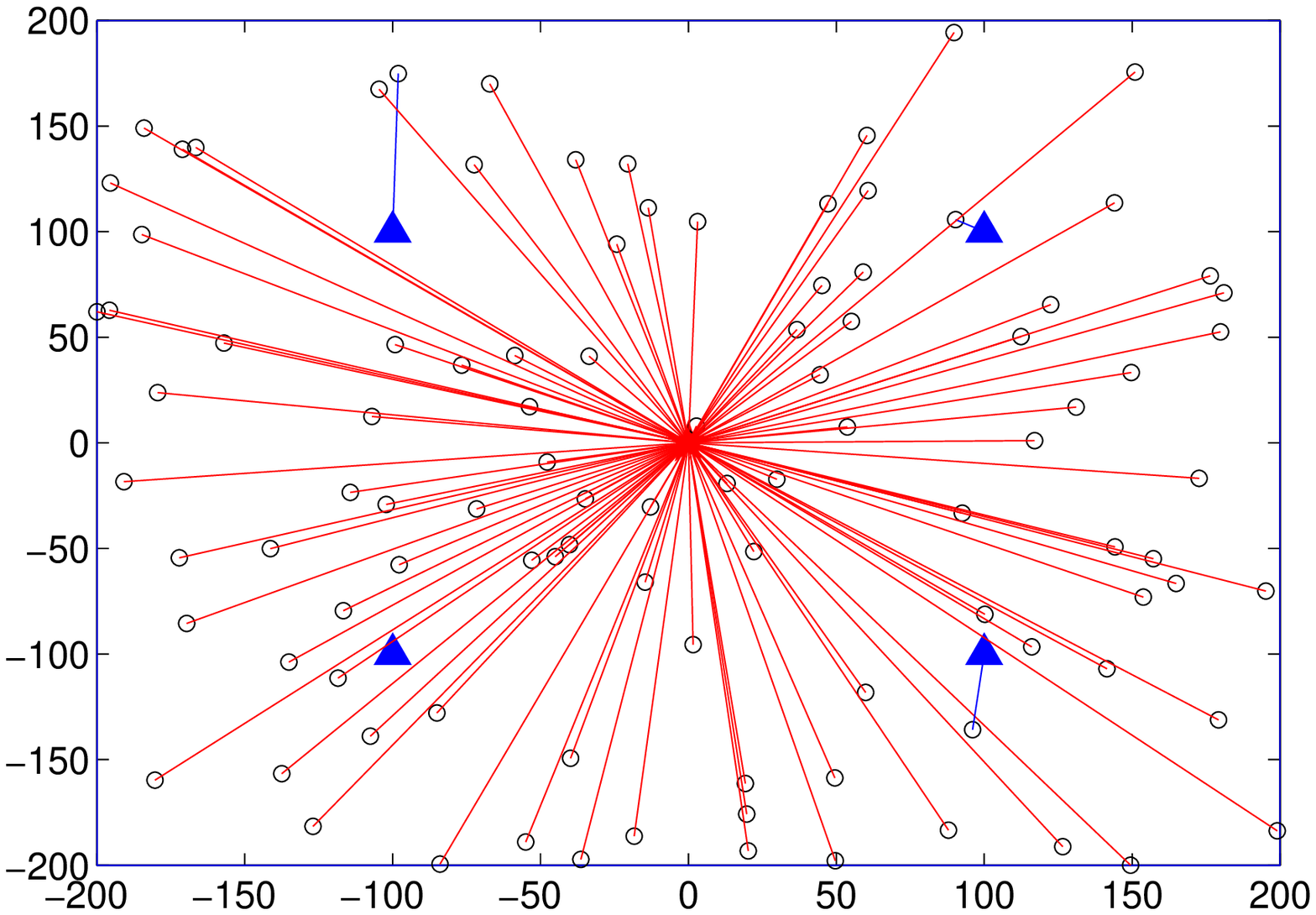}
   \label{fig:1002}
   }
   \subfigure[]{
  \includegraphics[width=8.5cm]{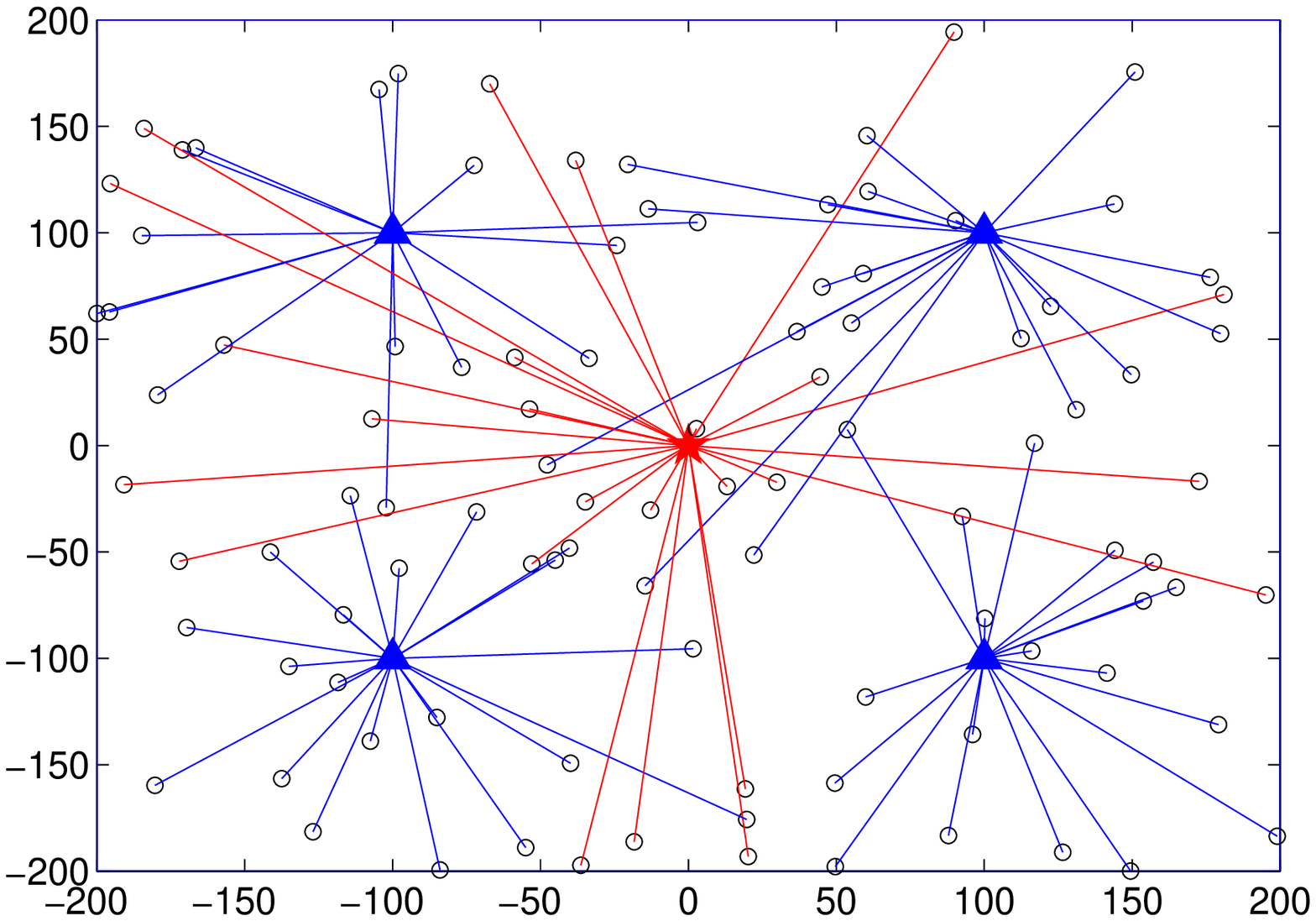}
   \label{fig:1003}
   }
 \caption[]{A realization of load distribution under different associations in a two-tier network: max-received-power association in Figure \ref{fig:1001}; max-SNR association in Figure \ref{fig:1002}; max-selection-metric association in Figure \ref{fig:1003}.  
}
\label{fig:1}
\end{figure}

This proposed scheme equips the networks with a powerful access control method that facilitates load balancing. Further, as will be shown later, the proposed scheme is throughput optimal with a carefully selected connection threshold $\mu$. Note that it is also up to the BS to schedule the radio resources for each connected UE. In particular, the actual scheduling probability or the allocated radio resources of a UE can deviate from the expected scheduling probability broadcast as part of the system information.

\section{System Model and Problem Formulation}
\label{sec:model}

In this section, we describe a system model and formulate a throughput optimization problem for the initial cell search and selection scheme proposed in Section \ref{sec:proposed}.

\subsection{Synchronization}
\label{subsec:sync}

In 5G networks, UE needs to determine the transmit spatial signature of a candidate BS and its corresponding receive spatial signature. We assume that UE learns the beamforming directions by detecting synchronization signals. Specifically, BSs periodically broadcast known synchronization signals with a period of $T_{\textrm{syn}}$ seconds. To enable the detection of beamforming directions, each BS cycles through its set of transmit spatial signatures, and each UE also cycles through its set of receive spatial signatures. In a nutshell, denoting by $L$ the number of possible transmit-receive spatial signature pairs, each beamforming scan cycle takes $L T_{\textrm{syn}}$ seconds.

The described synchronization model is general and flexible. For example, assume that each BS and each UE are equipped with an antenna array with $N_{\textrm{bs}}$ antennas and $N_{\textrm{ue}}$ antennas, respectively. Accordingly, $N_{\textrm{bs}}$ and $N_{\textrm{ue}}$ orthogonal beamspace directions are available at the BS and the UE, respectively. If both the BS and the UE scan over all the orthogonal beamspace directions, $L = N_{\textrm{bs}} N_{\textrm{ue}}$. If the BS scans over all the orthogonal beamspace directions but the UE uses an omni-directional or fixed antenna pattern, $L = N_{\textrm{bs}}$. On the contrary, $L = N_{\textrm{ue}}$ if the UE scans over all the orthogonal beamspace directions but the BS uses an omni-directional or fixed antenna pattern. If both the BS and the UE use omni-directional or fixed antenna patterns, $L = 1$.

Synchronization channels are usually designed to enable high detection rate at very low SNR. Therefore, for simplicity we assume that each UE is able to obtain synchronization and determine the beamforming directions after a beamforming scan cycle.

\subsection{Extracting Cell Load Information}
\label{subsec:load}

We assume that BSs periodically broadcast the system information block containing  cell load information with a period of $T_{\textrm{sib}}$ seconds. Without loss of generality, we assume $T_{\textrm{sib}} \geq T_{\textrm{syn}}$; the analysis in this paper can be straightforwardly extended to the case $T_{\textrm{sib}} < T_{\textrm{syn}}$. The period $T_{\textrm{sib}}$ is further chosen such that 
$T_{\textrm{sib}}/T_{\textrm{syn}}$ is an integer, facilitating the alignment of synchronization signals and system information broadcast channels. In particular, each system information broadcast channel can be positioned right after a synchronization signal. This minimizes the waiting time between a synchronization signal and a system information broadcast channel, leading to reduced latency of initial cell search and selection.


Similar to synchronization channels, system information broadcast channels are usually designed to cover UEs with low SNR. Therefore, for simplicity we assume that each UE is able to decode the system information block and extract cell load information in one shot.

\subsection{Random Access Procedure and Data Communication}

After UE finishes initial cell search and selection, it can initiate the random access procedure by transmitting a random access preamble to the selected BS. Once the random access procedure is completed, a connection between the UE and the BS is established, after which data communication may be scheduled. 

We assume for simplicity that the random access procedure takes a fixed time of $T_{\textrm{ra}}$ seconds. We focus on downlink data communication and assume that it lasts for a time duration of $T_{\textrm{data}}$ seconds.

\subsection{Problem Formulation}

Consider a full communication period, during which a UE searches for $n$ cells in the initial cell search and then selects a cell for the random access and data communication. Let $T_n$ be the duration of the period. It follows that
\begin{align}
 T_n  = \sum_{i=1}^n \left(Y_i + (L-1) T_{\textrm{syn}}  + Z_i \mathbb{I} ( \snr_i \geq \Gamma ) \right) + T_{\textrm{ra}} + T_{\textrm{data}} ,
 \label{eq:01}
\end{align}
where $Y_i$ denotes the duration between the time instant that UE starts searching for cell $i$ and the time instant that UE finds the first synchronization signal of cell $i$, and
$Z_i$ denotes the duration between the start of the last synchronization signal that the UE searches in cell $i$ and the start of the following system information broadcast channel that the UE tries to decode in cell $i$. 

We assume that the selection metrics $\{R_i\}$ defined in (\ref{eq:1}) are  independent and identically distributed (i.i.d.). Denote by $R$ the generic random variable for $\{R_i\}$ with cumulative distribution function $F_R(x)$. We further assume that the second moment of $R$ exists. After searching $n$ cells, the UE may select the best cell for the data communication. Accordingly, UE may expect that the amount of downlink bits that it can receive equals
\begin{align}
U_n =W T_{\textrm{data}} \max_{i=1,...,n}  R_{i} \mathbb{I} ( \snr_i \geq \Gamma ) ,
\label{eq:02}
\end{align}
where $W$ denotes the channel bandwidth. Searching for more cells increases the probability of finding a better cell, since $U_n$ is monotonically increasing with $n$. However, increasing the number $n$ of searched cells also increases the duration of the overall communication period $T_n$. The tradeoff naturally raises the question: How many cells should the UE search before it stops searching and starts the random access procedure and data communication? 

If the same cell search and selection rule is used across $m$ communication periods, the total number of downlink bits that the UE can receive equals
$\sum_{k=1}^m U_{n(k)}$, where $n(k)$ denotes the number of cells searched in the $k$-th period. Accordingly, the duration of the $m$ periods equals $\sum_{k=1}^m T_{n(k)}$. Therefore, the throughput (bit/s) equals $\sum_{k=1}^m U_{n(k)}/\sum_{k=1}^m T_{n(k)}$. Letting $m\to\infty$, by the law of large numbers, the ergodic throughput equals ${ \mathbb{E}[ U_N ] }/{\mathbb{E}[ T_N ]}$.

Denoting the set of admissible cell search and selection rules by 
\begin{align}
\mathcal{C}=\{N \in \mathbb{N}^+: \mathbb{E}[T_N] < +\infty\},
\end{align}
our objective is to find an optimal stopping rule $N^\star$ under the proposed cell search and selection scheme to obtain the maximum ergodic throughput $\lambda^\star$. Mathematically, the throughput optimization problem is written as follows:
\begin{align}
\lambda^\star \triangleq \sup_{N \in \mathcal{C}} \frac{ \mathbb{E}[ U_N ] }{\mathbb{E}[ T_N ]}.
\label{eq:3}
\end{align}

\section{Throughput Optimal Stopping of Initial Cell Search and Selection}
\label{sec:optimal}

\subsection{Threshold Policy Achieves the Optimal Throughput} 
 
In this section, we characterize the throughput optimal stopping strategy $N^\star$ and the attained maximum throughput $\lambda^\star$. To this end, we first derive the expected duration of a typical communication period in Lemma \ref{lem:1}.

\begin{lem}
The expected duration $\mathbb{E}[T_n]$ of a communication period with $n$ searched cells is given by
\begin{align}
\mathbb{E} [ T_n ] = & n \left(L - \frac{1}{2} \right) T_{\textrm{syn}} +  \frac{n}{2}(T_{\textrm{sib}} - T_{\textrm{syn}})    \mathbb{P} ( \snr\geq \Gamma )   \notag \\
& + T_{\textrm{ra}} + T_{\textrm{data}} .
\label{eq:6}
\end{align}
\label{lem:1}
\end{lem}
\begin{proof}
By the expression (\ref{eq:01}) of $T_n$ and the linearity of expectation, we have that
\begin{align}
 \mathbb{E} [T_n]  &= \sum_{i=1}^n \left( \mathbb{E} [Y_i] + (L-1) T_{\textrm{syn}}  + \mathbb{E} [ Z_i \mathbb{I} ( \snr_i \geq \Gamma ) ] \right) \notag \\
 &\quad \quad \quad \quad + T_{\textrm{ra}} + T_{\textrm{data}} \notag \\
 &= n(L-1) T_{\textrm{syn}}  + \sum_{i=1}^n \left( \mathbb{E} [Y_i ] +  \mathbb{E} [ Z_i ]  \mathbb{P} [  \snr_i \geq \Gamma  ] \right)  \notag \\
 &\quad \quad \quad \quad + T_{\textrm{ra}} + T_{\textrm{data}} ,
 \label{eq:4}
\end{align}
where in the last equality we have used the fact that $Z_i$ and $\snr_i$ are independent. By the assumptions in Section \ref{subsec:sync}, it is clear that
\begin{align}
\mathbb{E} [Y_i] = \frac{T_{\textrm{syn}}}{2}  , \quad\forall i = 1, ..., n.
\label{eq:5}
\end{align}
By the assumption in Section \ref{subsec:load}, we have that 
\begin{align}
\mathbb{P} ( Z_i = j T_{\textrm{syn}} ) = \frac{T_{\textrm{syn}}}{T_{\textrm{sib}}}, \quad j = 0, ..., \frac{T_{\textrm{sib}}}{T_{\textrm{syn}}} - 1 .
\end{align}
It follows that
\begin{align}
\mathbb{E} [Z_i] = \frac{T_{\textrm{sib}} - T_{\textrm{syn}}}{2} , \quad\forall i = 1, ..., n.
\label{eq:05}
\end{align}
Plugging (\ref{eq:5}) and (\ref{eq:05})  into (\ref{eq:4}) yields (\ref{eq:6}).
\end{proof}

While Lemma \ref{lem:1} characterizes the expected duration of a communication period, during which the UE searches $n$ cells and selects a cell for the random access and data communication, different stopping rules yield different numbers of searched cells that possibly vary across communication periods. With Lemma \ref{lem:1}, we are now in a position to derive the optimal stopping rule in Proposition \ref{pro:1}.

\begin{pro}
The optimal stopping rule for the throughput maximization problem (\ref{eq:3}) is given by
\begin{align}
& N^\star =\min \left\{ n \geq 1:  U_n   \geq  \lambda^\star ( T_{\textrm{ra}} + T_{\textrm{data}} ) \right\},
\label{eq:12}
\end{align}
where $\lambda^\star$ is the unique maximum throughput. Further, $\lambda^\star$ is the solution to the following fixed point equation:
\begin{align}
&\mathbb{E} \left[ ( U_n  - \lambda  ( T_{\textrm{ra}} + T_{\textrm{data}} )  )^+  \right] \notag \\
&=  \lambda  \left( (L - \frac{1}{2}) T_{\textrm{syn}}  + \frac{T_{\textrm{sib}} - T_{\textrm{syn}}}{2} \mathbb{P} ( \snr_n \geq \Gamma ) \right)   ,
\label{eq:11}
\end{align}
where $(x)^+ \triangleq \max(x, 0)$.
\label{pro:1}
\end{pro}
\begin{proof}
We first solve the associated ordinary optimal stopping problem:
\begin{align}
V( \lambda ) &\triangleq  \sup_{N \in \mathcal{C}}   \mathbb{E} [U_N - \lambda T_N ]  ,
\label{eq:2} 
\end{align}
where $\lambda$ is an arbitrary positive number.
The duration of searching for cell $1$ equals
$
Y_1 + (L-1) T_{\textrm{syn}}  + Z_1 \mathbb{I} ( \snr_1 \geq \Gamma ).
$
If the UE stops searching and selects cell $1$ for the random access and data communication, it can receive $U_1 = R_{1} \mathbb{I} ( \snr_1 \geq \Gamma )$ information bits. If instead the UE continues to search for more cells from this point, the $U_1$ bits are not transmitted and the time $Y_1 + (L-1) T_{\textrm{syn}}  + Z_1 \mathbb{I} ( \snr_1 \geq \Gamma )$ has passed. In the search of cell $2$, the problem starts over again, implying that the problem is invariant in time. 
Generalizing this to stage $n$, if the UE stops searching and selects the best cell from the $n$ searched cells for the random access and data communication, it obtains a utility of $U_n - \lambda T_n$. If the UE continues to search for more cells, it obtains a utility of $V(\lambda)- \lambda (T_{n-1} + Y_n + (L-1) T_{\textrm{syn}}  + Z_n \mathbb{I} ( \snr_n \geq \Gamma ) )$. By the optimality equation of dynamic programming \cite{bertsekas1995dynamic},
\begin{align}
&V (\lambda) -\lambda T_{n-1}  = \mathbb{E} [ \max ( U_n - \lambda T_n,  V (\lambda)    \notag \\
& -\lambda   (T_{n-1} + Y_n + (L-1) T_{\textrm{syn}}  + Z_n \mathbb{I} ( \snr_n \geq \Gamma ) ) ) ] .
\label{eq:07}
\end{align}
To obtain the maximum utility $V(\lambda)$, the UE can stop searching once the currently achievable utility is not less than the maximum expected utility that is obtained with continuing. In other words, the UE can stop searching if 
\begin{align}
& U_n -  \lambda T_n \geq  V (\lambda)    \notag \\
&- \lambda   (T_{n-1} + Y_n + (L-1) T_{\textrm{syn}}  + Z_n \mathbb{I} ( \snr_n \geq \Gamma ) ).
\label{eq:09}
\end{align}
Adding $\lambda T_{n}$ to both sides of (\ref{eq:09}) yields that
\begin{align}
U_n  \geq  V (\lambda) + \lambda  ( T_{\textrm{ra}} + T_{\textrm{data}} ).
\end{align}
Therefore, the optimal stopping strategy for the associated ordinary optimal stopping problem is given by
\begin{align}
N^\star (\lambda)  = \min \{ n \geq 1: U_n  \geq  \lambda  ( T_{\textrm{ra}} + T_{\textrm{data}} ) + V (\lambda) \} .
\label{eq:9}
\end{align}

Next we characterize $V (\lambda)$. Adding $\lambda T_{n-1}$ to both sides of (\ref{eq:07}) yields that
\begin{align}
& V (\lambda)  = \mathbb{E} [ \max ( U_n - \lambda (T_n - T_{n-1}),  V (\lambda)    \notag \\
&\quad   -\lambda   (Y_n + L T_{\textrm{syn}}  + Z_n \mathbb{I} ( \snr_n \geq \Gamma ) ) ) ]  \notag \\
&= \mathbb{E} [ \max ( U_n - \lambda ( T_{\textrm{ra}} + T_{\textrm{data}} ),  V (\lambda) ) ]   \notag \\
&\quad  -  \mathbb{E} [ \lambda   (Y_n + (L-1) T_{\textrm{syn}}  + Z_n \mathbb{I} ( \snr_n \geq \Gamma ) )  ]  \notag \\
&= \mathbb{E} [ \max ( U_n - \lambda ( T_{\textrm{ra}} + T_{\textrm{data}} ),  V (\lambda) ) ] \notag \\
&\quad  -   \lambda  \left( (L - \frac{1}{2}) T_{\textrm{syn}}  + \frac{T_{\textrm{sib}} - T_{\textrm{syn}}}{2} \mathbb{P} ( \snr_n \geq \Gamma ) \right)  ,
\label{eq:7}
\end{align}
where we have plugged $\mathbb{E} [Y_i]$ and $\mathbb{E} [Z_i] $ (c.f. (\ref{eq:5}) and (\ref{eq:05})) into the last equality. Rearranging the terms in (\ref{eq:7}) yields that
\begin{align}
&\mathbb{E} \left[ ( U_n  - \lambda  ( T_{\textrm{ra}} + T_{\textrm{data}} ) - V (\lambda)  )^+  \right] \notag \\
&= \lambda  \left( (L - \frac{1}{2}) T_{\textrm{syn}}  + \frac{T_{\textrm{sib}} - T_{\textrm{syn}}}{2} \mathbb{P} ( \snr_n \geq \Gamma ) \right)   .
\label{eq:8}
\end{align}

By Theorem 1, Chapter 6 in \cite{ferguson2012optimal}, we know that  $N^\star$ is an optimal stopping rule that attains the maximum throughput $\lambda^\star$ in the throughput optimization problem (\ref{eq:3}) if and only if $N^\star$ is an optimal stopping rule for the ordinary optimal stopping problem (\ref{eq:2}) with $\lambda=\lambda^\star$ and $V(\lambda^\star) = 0$. Plugging $V(\lambda) = 0$ into (\ref{eq:8})  yields (\ref{eq:11}). Letting $V(\lambda) = 0$ in (\ref{eq:9}) yields the optimal stopping strategy in (\ref{eq:12}). 

The left side of (\ref{eq:11}) is continuous in $\lambda$ and decreasing from $\mathbb{E}[U_n^+]$ to zero, while the right side of (\ref{eq:11}) is continuous in $\lambda$ and increasing from $0$ to $+\infty$. Hence, there is a unique solution $\lambda^\star$. This completes the proof.
\end{proof}

Proposition \ref{pro:1} implies that the optimal stopping rule is a pure threshold policy: The initial cell search and selection process stops once the maximum of the expected numbers of downlink bits of the $n$ searched cells exceeds an optimized threshold, i.e.,
\begin{align}
U_n=W T_{\textrm{data}} \max_{i=1,...,n}  R_{i} \mathbb{I} ( \snr_i \geq \Gamma ) \geq \lambda^\star ( T_{\textrm{ra}} + T_{\textrm{data}} ) .
\end{align}
In fact, we can tighten the conclusion and show that the initial cell search and selection process can stop based on the currently examined cell only, as summarized in Proposition \ref{pro:2}.
\begin{pro}
The optimal stopping rule for the throughput maximization problem (\ref{eq:3}) is given by
\begin{align}
& N^\star =\min \left\{ n \geq 1:  R_n \mathbb{I} ( \snr_n \geq \Gamma )  \geq  \frac{\lambda^\star  ( T_{\textrm{ra}} + T_{\textrm{data}} ) }{W T_{\textrm{data}}} \right\},
\label{eq:012}
\end{align}
where $\lambda^\star$ is the unique maximum throughput. Further, $\lambda^\star$ is the solution to the following fixed point equation:
\begin{align}
&\mathbb{E} \left[ ( W T_{\textrm{data}} R_n \mathbb{I} ( \snr_n \geq \Gamma ) - \lambda  ( T_{\textrm{ra}} + T_{\textrm{data}} )  )^+  \right] \notag \\
&=  \lambda  \left( (L - \frac{1}{2}) T_{\textrm{syn}}  + \frac{T_{\textrm{sib}} - T_{\textrm{syn}}}{2} \mathbb{P} ( \snr_n \geq \Gamma ) \right)  .
\label{eq:011}
\end{align}
\label{pro:2}
\end{pro}
\begin{proof}
We first claim that the rule (\ref{eq:12}) is equivalent to the rule
$
\hat{N} = \min \{ n \geq 1: \hat{R}_n \geq   \rho  \},
$
where for notational simplicity we define 
\begin{align}
\hat{R}_n = W T_{\textrm{data}} R_n \mathbb{I} ( \snr_n \geq \Gamma ) 
\label{eq:0111}
\end{align}
and $\rho = \lambda^\star ( T_{\textrm{ra}} + T_{\textrm{data}} )$. This can be shown by induction. Clearly, at stage $1$ the stopping rule (\ref{eq:12})  and $\hat{N}$ are the same because $U_1 = \hat{R}_1$. In particular, if $U_1 = \hat{R}_1 \geq \rho $, both the rule (\ref{eq:12}) and $\hat{N}$ call for stopping.  If $U_1 = \hat{R}_1 < \rho $, both the rule (\ref{eq:12}) and $\hat{N}$ call for continuing to stage $2$. At stage $2$, if $U_2 = \max (\hat{R}_1, \hat{R}_2) \geq \rho$, then $U_2  = \hat{R}_2$ because $\hat{R}_1 < \rho$ by induction. It follows that both the rule (\ref{eq:12}) and $\hat{N}$ call for stopping.  If $U_2 = \max (\hat{R}_1, \hat{R}_2) < \rho$, then $\hat{R}_2 < \rho$. Thus, both the rule (\ref{eq:12}) and $\hat{N}$ call for continuing to stage $3$. Repeating this argument for stages $3,4,...$, we can see that the rule (\ref{eq:12}) and $\hat{N}$ are  the same stopping rules. Further, it is obvious that $\hat{N}$ is equivalent to the rule (\ref{eq:012}). 

Now we have shown that the optimal initial cell search and selection rule is to select the first cell satisfying $\hat{R}_n \geq \rho$. In particular, it is not necessary to recall any of the previously scanned cells, and the fixed point equation (\ref{eq:011}) can be derived along the same line of the proof of Proposition \ref{pro:1}.
\end{proof}

Proposition \ref{pro:2} implies that not only is the optimal stopping rule a pure threshold policy but also the optimal stopping is based on the currently examined cell $n$ only, i.e.,
\begin{align}
W T_{\textrm{data}}  R_{n} \mathbb{I} ( \snr_n \geq \Gamma ) \geq \lambda^\star ( T_{\textrm{ra}} + T_{\textrm{data}} ) .
\end{align}
In other words, it is not necessary to recall any previously scanned cells: Simply select the cell scanned at the stopping stage $N^\star$. This is a desirable feature in practical systems. In particular, due to e.g., the time varying radio environment, UE mobility, and clock drift, the synchronization with an earlier cell and/or the extracted cell load information might become outdated. Choosing the most recently scanned cell avoids such nuisances.

\subsection{An Example with Binary Selection Metric}

Now we have shown that the proposed initial cell search and selection scheme detailed in Section \ref{sec:proposed} is throughput optimal if the connection threshold $\mu$ is chosen to be
$\frac{ T_{\textrm{ra}} + T_{\textrm{data}}  }{W T_{\textrm{data}}} \lambda^\star$. The threshold however does not admit a closed-form solution and involves solving the fixed point equation (\ref{eq:011}). In what follows, to gain insights  we consider a special case in Corollary \ref{cor:1} where the selection metric of each cell only takes two values.

\begin{cor}
Assume that $\snr_i \geq \Gamma, \forall i$, and that $R = R_{\max}$ with probability $q$ and $R = 0$ with probability $1-q$. The maximum throughput $\lambda^\star $ is given by
\begin{align}
\lambda^\star  = \frac{ q W T_{\textrm{data}}  R_{\max} }{ (L - \frac{1}{2}) T_{\textrm{syn}}  + \frac{T_{\textrm{sib}} - T_{\textrm{syn}}}{2} + q (  T_{\textrm{ra}} + T_{\textrm{data}} )  } 
\label{eq:25}
\end{align}
with the optimal stopping strategy given by $N^\star   = \min \{ n \geq 1: R_n  \geq \frac{1}{1+\phi} R_{\max}  \}$, where 
\begin{align}
\phi = \frac{(L - \frac{1}{2}) T_{\textrm{syn}}  + \frac{T_{\textrm{sib}} - T_{\textrm{syn}}}{2}}{ q (  T_{\textrm{ra}} + T_{\textrm{data}} )  }  .
\label{eq:24}
\end{align}
\label{cor:1}
\end{cor}

Several remarks on the results in Corollary \ref{cor:1} are in order.

\textbf{Remark 1.} The numerator $q W T_{\textrm{data}}  R_{\max}$ in (\ref{eq:25}) is the expected number of information bits that UE can receive in the downlink in a communication period. The denominator in (\ref{eq:25}) is the expected duration of a communication period including the expected synchronization time $(L - \frac{1}{2}) T_{\textrm{syn}}$, the expected time $\frac{T_{\textrm{sib}} - T_{\textrm{syn}}}{2}$ of extracting cell load information, the expected time $q T_{\textrm{ra}}$ of the random access procedure, and the expected time $q T_{\textrm{data}}$ of data communication. Therefore, with the optimal stopping strategy the UE makes the right decision in initial cell search and selection and achieves the optimal throughput given in (\ref{eq:25}).

\textbf{Remark 2.} The optimal stopping rule is a pure threshold policy: the UE stops searching and selects the currently scanned cell if $R_n/ R_{\max} \geq  \frac{1}{1+\phi} $. We can see that the  the threshold is determined by the ratio $\phi$  of the expected cell search time $((L - \frac{1}{2}) T_{\textrm{syn}}  + \frac{T_{\textrm{sib}} - T_{\textrm{syn}}}{2})/q$ and the expected time $T_{\textrm{ra}} + T_{\textrm{data}}$ used in the random access and data communication. Intuitively, the longer the data communication, the higher the threshold, i.e., the UE is more cautious in cell selection and is willing to search for more cells. In contrast, the longer the expected cell search time, the lower the threshold. This is because the overhead of searching for a cell becomes higher. As a result, the UE should decrease its threshold and stops earlier. 

\textbf{Remark 3.} The maximum throughput can be written as
\begin{align}
\lambda^\star  = \frac{1}{1+\phi} \cdot  \frac{T_{\textrm{data}} }{ T_{\textrm{ra}} + T_{\textrm{data}}  } W  R_{\max} ,
\end{align}
which shows the dependency of the maximum throughput on the ratio $\phi$ of the expected cell search time $((L - \frac{1}{2}) T_{\textrm{syn}}  + \frac{T_{\textrm{sib}} - T_{\textrm{syn}}}{2})/q$ and the expected time $T_{\textrm{ra}} + T_{\textrm{data}}$ used in the random access and data communication. Intuitively, the higher the ratio $\phi$, the lower the maximum throughput, due to the increased cell search overhead.

\textbf{Remark 4.} Note that the threshold of the optimal stopping (\ref{eq:012}) may not be unique, though the maximum throughput $\tilde{\lambda}^\star$ is unique. The binary selection metric taking either value $R_{\max}$ or $0$ in Corollary \ref{cor:1} is one such example. In particular, any value in $(0, R_{\max}]$ can be used as a threshold to achieve the maximum throughput.


\subsection{An Alternative Characterization of the Optimal Throughput}

Proposition \ref{pro:2} characterizes the optimal throughput via the fixed point equation (\ref{eq:011}). In this section, we derive an alternative characterization of the optimal throughput. The alternative characterization will pave the way for developing an iterative algorithm to compute the solution to the fixed point equation (\ref{eq:011}).

Denote by $\hat{R}$ the generic random variable for $\{\hat{R}_n\}$ defined in (\ref{eq:0111}) and by $F_{\hat{R}}(x)$ the cumulative distribution function of $\hat{R}$. The following Corollary \ref{cor:2} readily follows.
\begin{cor}
With the throughput optimal stopping in initial cell search and selection, the following results hold. 
\begin{enumerate}
\item The stopping time ${N}^\star$ is geometrically distributed with parameter $1-F_{\hat{R}}({\lambda}^\star (T_{\textrm{ra}} + T_{\textrm{data}}) )$.
\item The distribution of the stopped random variable $U_{{N}^\star}$ is given by 
\begin{align}
F_{U_{{N}^\star}} (x) = \frac{ F_{\hat{R}}(x) - F_{\hat{R}}({\lambda}^\star (T_{\textrm{ra}} + T_{\textrm{data}}) )  }{ 1-F_{\hat{R}}({\lambda}^\star (T_{\textrm{ra}} + T_{\textrm{data}})) }
\end{align}
with $x\geq {\lambda}^\star (T_{\textrm{ra}} + T_{\textrm{data}})$.
\end{enumerate}
\label{cor:2}
\end{cor}
\begin{proof}
From Proposition \ref{pro:2}, we know that UE stops searching and selects cell $n$ if $\hat{R}_n \geq {\lambda}^\star (T_{\textrm{ra}} + T_{\textrm{data}})$. It follows that the number ${N}^\star$ of searched cells   is geometrically distributed with parameter $1-F_{\hat{R}}({\lambda}^\star (T_{\textrm{ra}} + T_{\textrm{data}}) )$. At the throughput optimal stopping time $N^\star$, $U_{{N}^\star} = \hat{R}_{N^\star}$. Further, the distribution of $\hat{R}_{N^\star}$ is a conditional distribution that results from restricting the domain of the distribution $\hat{R}$ to $x\geq {\lambda}^\star (T_{\textrm{ra}} + T_{\textrm{data}})$. This completes the proof.
\end{proof}

With Corollary \ref{cor:2}, we are now in a position to derive the alternative characterization of the optimal throughput in the following Proposition \ref{pro:3}.
\begin{pro}
With the throughput optimal stopping in initial cell search and selection, the maximum throughput $\lambda^\star $ is given by
\begin{align}
\lambda^\star = \frac{ \int_{{\lambda}^\star (T_{\textrm{ra}} + T_{\textrm{data}}) }^\infty  x  \dint F_{\hat{R}} (x)  }{ \eta   + (T_{\textrm{ra}} + T_{\textrm{data}}) (1-F_{\hat{R}}({\lambda}^\star (T_{\textrm{ra}} + T_{\textrm{data}}) ) } ,
\label{eq:29}
\end{align}
where $\eta \triangleq \left(L - \frac{1}{2} \right) T_{\textrm{syn}}  +  \frac{1}{2}(T_{\textrm{sib}} - T_{\textrm{syn}})    \mathbb{P} ( \snr\geq \Gamma )  $.
\label{pro:3}
\end{pro}
\begin{proof}
By Lemma \ref{lem:1}, we have
\begin{align}
&\mathbb{E} [ T_{N^\star} ] =  \mathbb{E} [N^\star]  \left(L - \frac{1}{2} \right) T_{\textrm{syn}} \notag \\
& +  \frac{\mathbb{E} [N^\star] }{2}(T_{\textrm{sib}} - T_{\textrm{syn}})    \mathbb{P} ( \snr\geq \Gamma )    + T_{\textrm{ra}} + T_{\textrm{data}} .
\label{eq:26}
\end{align}
By the first result in Corollary \ref{cor:2}, we have
\begin{align}
& \mathbb{E} \left[  N^\star \right]  =  \frac{1}{ 1-F_{\hat{R}}({\lambda}^\star (T_{\textrm{ra}} + T_{\textrm{data}}) )  }  .
\label{eq:27}
\end{align} 
By the second result in Corollary \ref{cor:2}, we have
\begin{align}
&\mathbb{E}[U_{{N}^\star}] = \int_{{\lambda}^\star (T_{\textrm{ra}} + T_{\textrm{data}})  }^\infty  x  \dint F_{U_{{N}^\star}} (x)  \notag \\
&= \frac{ 1  }{ 1-F_{\hat{R}}({\lambda}^\star (T_{\textrm{ra}} + T_{\textrm{data}})) } \int_{{\lambda}^\star (T_{\textrm{ra}} + T_{\textrm{data}}) }^\infty  x  \dint F_{\hat{R}} (x)  .
\label{eq:28}
\end{align}
Plugging (\ref{eq:26}), (\ref{eq:27}), and (\ref{eq:28}) into ${\lambda}^\star = \frac{\mathbb{E}[U_{{N}^\star}]}{ \mathbb{E} [ T_{{N}^\star } ] } $ yields (\ref{eq:29}). 
\end{proof}

Proposition \ref{pro:3} provides an alternative fixed point equation (\ref{eq:29}) whose solution is the maximum throughput. It also suggests one possible numerical iterative algorithm to solve for $\lambda^\star$. Denote by $t$ the iteration index. Replacing the $\lambda^\star$ on the left hand side of (\ref{eq:29}) by $\lambda[t+1] $ and the $\lambda^\star$ on the right hand side of (\ref{eq:29}) by $\lambda[t]$ yields that
\begin{align}
\lambda[t+1] = \frac{ \int_{{\lambda}[t] (T_{\textrm{ra}} + T_{\textrm{data}}) }^\infty  x  \dint F_{\hat{R}} (x)  }{ \eta   + (T_{\textrm{ra}} + T_{\textrm{data}}) (1-F_{\hat{R}}({\lambda}[t] (T_{\textrm{ra}} + T_{\textrm{data}}) ) } .
\label{eq:30}
\end{align}
This iterative method is in essence a variation of Newton's method with all iterations using a unit step size. The following Proposition \ref{pro:4} formally establishes the convergence of the iterative equation (\ref{eq:30}).

\begin{pro}
For any initial value $\lambda [0] > 0$ , the sequence $\{\lambda[t]\}$ generated by the iterative equation (\ref{eq:30}) converges to the maximum throughput $\lambda^\star$. 
\label{pro:4}
\end{pro}
\begin{proof}
For notational simplicity, denote by
\begin{align}
h(z) = \frac{ \int_{z (T_{\textrm{ra}} + T_{\textrm{data}}) }^\infty  x  \dint F_{\hat{R}} (x)  }{ \eta   + (T_{\textrm{ra}} + T_{\textrm{data}}) (1-F_{\hat{R}}(z (T_{\textrm{ra}} + T_{\textrm{data}}) ) }.
\end{align}
By definition, $h(0) > 0$, ${\lambda}^\star = \max_{\lambda} h(\lambda)$, and $\lambda^\star$ is the unique maximum value. It follows that  $\lambda \leq h(\lambda)$ if $\lambda \leq {\lambda}^\star$, and $\lambda > h(\lambda)$ if $\lambda > {\lambda}^\star$. For any initial value $\lambda[0] > 0$, if $\lambda[0] > {\lambda}^\star$, $\lambda[0] > h(\lambda[0]) = \lambda[1]$. Since $\lambda[1] \leq {\lambda}^\star$, we may assume without loss of generality that $\lambda[0] \leq {\lambda}^\star$. Further, we have $\lambda[0] \leq   h( \lambda[0]  ) = \lambda[ 1 ]$ and $h( \lambda[0]  ) \leq h( {\lambda}^\star  ) = {\lambda}^\star$. It follows that $\lambda[0] \leq    \lambda[ 1 ] \leq {\lambda}^\star$. By induction,  it can be seen that $\{\lambda[t]\}$ is monotonically non-decreasing and is bounded by ${\lambda}^\star$.  Therefore, $\{\lambda[t]\}$ converges to some limiting point $\lambda_{\infty}$. In particular,
\begin{align}
& 0 = \lim_{t \to \infty} ( \lambda [t + 1] - \lambda[t] ) \notag \\
&= \lim_{t \to \infty} \left( \frac{ \int_{{\lambda}[t] (T_{\textrm{ra}} + T_{\textrm{data}}) }^\infty  x  \dint F_{\hat{R}} (x)  }{ \eta   + (T_{\textrm{ra}} + T_{\textrm{data}}) (1-F_{\hat{R}}({\lambda}[t] (T_{\textrm{ra}} + T_{\textrm{data}}) ) }  - \lambda[t] \right) \notag \\
&= \frac{ \int_{{\lambda}_{\infty} (T_{\textrm{ra}} + T_{\textrm{data}}) }^\infty  x  \dint F_{\hat{R}} (x)  }{ \eta   + (T_{\textrm{ra}} + T_{\textrm{data}}) (1-F_{\hat{R}}({\lambda}_{\infty} (T_{\textrm{ra}} + T_{\textrm{data}}) ) } - \lambda_{\infty} .
\end{align}
It follows that 
\begin{align}
\lambda_{\infty} = \frac{ \int_{{\lambda}_{\infty} (T_{\textrm{ra}} + T_{\textrm{data}}) }^\infty  x  \dint F_{\hat{R}} (x)  }{ \eta   + (T_{\textrm{ra}} + T_{\textrm{data}}) (1-F_{\hat{R}}({\lambda}_{\infty} (T_{\textrm{ra}} + T_{\textrm{data}}) ) } .
\end{align}
In other words, $\lambda_{\infty}$ is also the solution to the fixed point equation (\ref{eq:5}) whose solution is the maximum throughput $\lambda^\star$. Since the solution is unique, we must have $\lambda_{\infty} = \lambda^\star$.
\end{proof}

\section{Simulation Results}
\label{sec:sim}

In this section, we provide simulation results to demonstrate the analytical results and obtain insights into how the various system parameters affect the throughput optimal initial cell search and selection. In the simulation, the load value broadcast by each cell $i$ is $\beta_i = \frac{1}{M_i + 1}$, where recall $M_i$ is the number of active UEs served by BS $i$. The received SNR from BS $i$ equals $\snr_i = g \cdot L \cdot \snr_{\textrm{avg}}$, where $g$ models log-normal shadowing and $L$ models the beamforming gain. The specific parameters used are summarized in Table \ref{tab:sys:para} unless otherwise specified. 

\begin{table}
\centering
\begin{tabular}{|l||r|} \hline
Mean number of active UEs per cell: $M_i$ & $10$  \\ \hline 
Standard deviation of log-normal shadowing & $7$ dB  \\ \hline 
Number of beamforming pairs: $L$ & $64$ \\ \hline 
Average received SNR: $\snr_{\textrm{avg}}$ & $-10$ dB  \\ \hline 
SNR threshold: $\Gamma$ & $-10$ dB  \\ \hline 
Synchronization signals period: $T_{\textrm{syn}}$ & $0.005$ $s$  \\ \hline 
System information reading period: $T_{\textrm{sib}}$ & $0.01$ $s$  \\ \hline 
Random access time: $T_{\textrm{ra}}$ & $0.02$ $s$  \\ \hline 
Data communication time: $T_{\textrm{data}}$ & $10$ $s$  \\ \hline 
Channel bandwidth: $W$ & $1$ GHz  \\ \hline 
\end{tabular}
\caption{Simulation Parameters}
\label{tab:sys:para}
\end{table}

Figure \ref{fig:2} shows a sample trace of initial cell search and selection. As the number $n$ of searched cells increases, the cell search time $T_n$ increases. Note that the cell search time $T_n$ is not a simple linear function of the number $n$ of searched cells (though visually a linear relationship is shown in Figure \ref{fig:2}). In particular, the cell search of each cell may consist of two parts: synchronization and cell load information reading. If the received SNR is below the threshold $\Gamma$, UE does not proceed with reading the load information after synchronization. The exact relationship is given in (\ref{eq:01}) excluding the last two terms $T_{\textrm{ra}}$ and $T_{\textrm{data}}$.  Figure \ref{fig:2} also shows that the amount of downlink bits that UE may expect to receive is a non-decreasing function of the number $n$ of searched cells, which is intuitive. To sum up, Figure \ref{fig:2} illustrates the tradeoff in searching for more cells in initial cell search and selection: Scanning more cells increases the probability of finding a better cell but at the cost of more overhead time spent on cell search.

 \begin{figure}
 \centering
  \includegraphics[width=9cm]{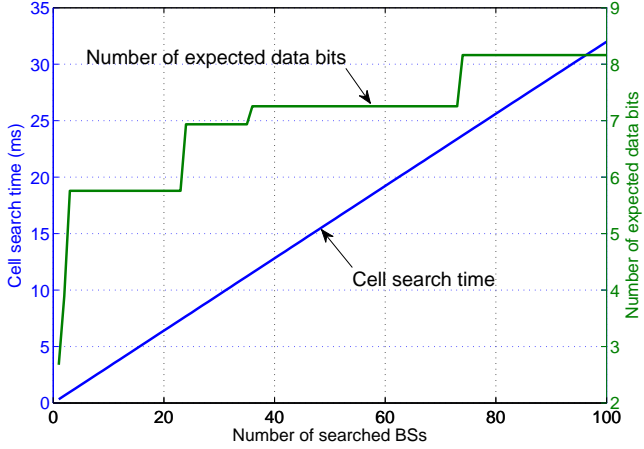}
 \caption[]{A sample trace of initial cell search and selection.}
\label{fig:2}
\end{figure}

In Figure \ref{fig:3}, we study how the throughput performance varies with the stopping selection metric threshold $R_i$ under different numbers $L$ of beamforming pairs. The data communication time is $T_{\textrm{data}}=10$ $s$, and the mean number of active UEs per cell equals $10$. For each number of beamforming pairs, Figure \ref{fig:3} clearly shows that there exists an optimal stopping threshold that achieves the maximum throughput. The maximum throughput increases when the number of beamforming pairs increases from $4$ to $16$ and to $64$, but it decreases when the number of beamforming pairs increases from $64$ to $256$. This is because there is a tradeoff when increasing the number of beamforming pairs.  Increasing the number of beamforming pairs increases the beamforming gain and in turn improves the received SNR. But increasing the number of beamforming pairs also increases the synchronization time spent on cell search since more beamforming pairs need to be scanned.

From Figure \ref{fig:3}, we can see that the optimal stopping selection metric threshold increases noticeably when the number of beamforming pairs increases from $4$ to $16$. So does the maximum throughput. This suggests that in this regime increasing beamforming gain is quite instrumental and much outweighs the cost of more time spent in cell search. In particular, UE can afford to search for more cells before camping on a cell and thus can set a higher stopping selection metric threshold. In contrast, when the number of beamforming pairs increases from $16$ to $64$ and to $256$, the optimal stopping selection metric threshold stays almost invariant and the maximum throughput does not change much. This suggests that in this regime the benefit from increasing the beamforming gain becomes saturated and is also offset by the increased overhead time in cell search.

 \begin{figure}
 \centering
  \includegraphics[width=9cm]{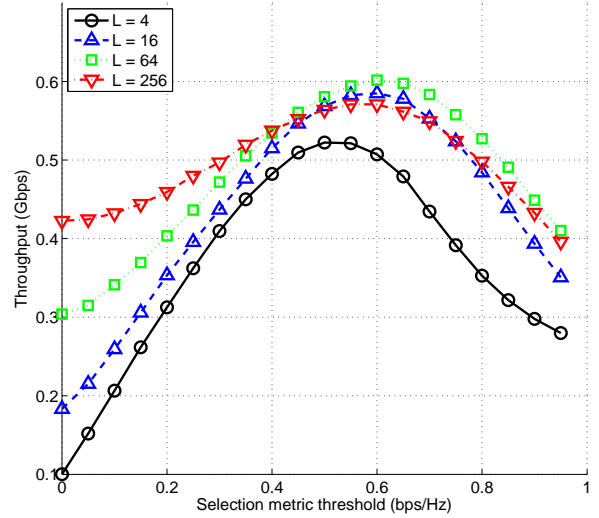}
 \caption[]{Ergodic throughput versus selection metric threshold under different numbers $L$ of beamforming pairs: $T_{\textrm{data}}=10$ $s$; mean number of active UEs per cell equals $10$.}
\label{fig:3}
\end{figure}

In Figure \ref{fig:4}, the setup is the same as in Figure \ref{fig:3} except that the data communication time is increased by $4$ times to $40$ $s$. Comparing Figure \ref{fig:4} to Figure \ref{fig:3}, we can see that for each number $L$ of beamforming pairs, the maximum throughput with $T_{\textrm{data}}=40$ $s$ in Figure \ref{fig:4} is larger than its corresponding part with $T_{\textrm{data}}=10$ $s$ in Figure \ref{fig:3}. This agrees with intuition: As the data communication time increases, the relative time overhead of cell search and random access becomes smaller, resulting in higher throughput. Further, for each number $L$ of beamforming pairs, the optimal selection metric threshold with $T_{\textrm{data}}=40$ $s$ in Figure \ref{fig:4} is larger than its counterpart with $T_{\textrm{data}}=10$ $s$ in Figure \ref{fig:3}. This is because the relative time overhead of cell search becomes smaller as the data communication time increases. As a result, UE can set a higher stopping selection metric threshold to search for more cells before camping on a cell.

 \begin{figure}
 \centering
  \includegraphics[width=9cm]{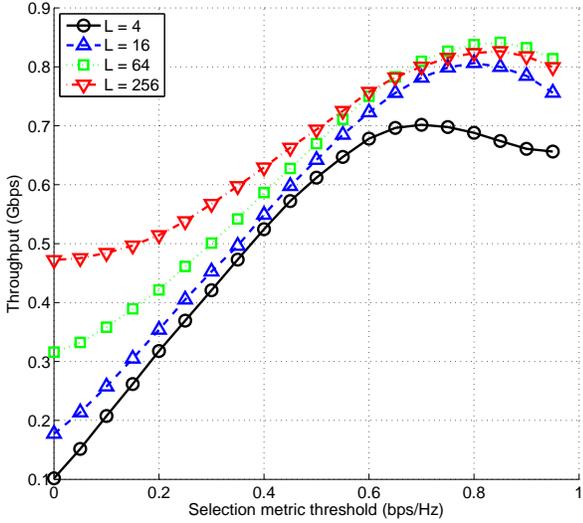}
 \caption[]{Ergodic throughput versus selection metric threshold under different numbers $L$ of beamforming pairs: $T_{\textrm{data}}=40$ $s$; mean number of active UEs per cell equals $10$.}
\label{fig:4}
\end{figure}

In Figure \ref{fig:5}, the setup is the same as in Figure \ref{fig:3} except that the mean number of active UEs per cell is decreased by $2$ times to $5$. Since $\beta_i = \frac{1}{M_i + 1}$, the expected scheduling probability of each cell $i$ is statistically larger in Figure \ref{fig:5}  than in Figure \ref{fig:3}. In other words, the load in Figure \ref{fig:5}  is statistically lighter than in Figure \ref{fig:3}. Comparing Figure \ref{fig:5} to Figure \ref{fig:3}, we can see that the  throughput values in Figure \ref{fig:5} are about twice as large as their counterparts in Figure \ref{fig:3}, agreeing with intuition. Accordingly,  for each number $L$ of beamforming pairs, the optimal selection metric threshold in Figure \ref{fig:4} is larger than its counterpart in Figure \ref{fig:3}. One interesting observation is that in Figure \ref{fig:5} the maximum throughput with $L=16$ is larger than the maximum throughput with $L=64$. The converse is true in Figure \ref{fig:3}. This suggests beamforming gain is more instrumental in improving the throughput performance when the load is heavier. When the load is light,  the throughput values are high since UE has access to more radio resources, and thus the beamforming gain becomes less important.

 \begin{figure}
 \centering
  \includegraphics[width=9cm]{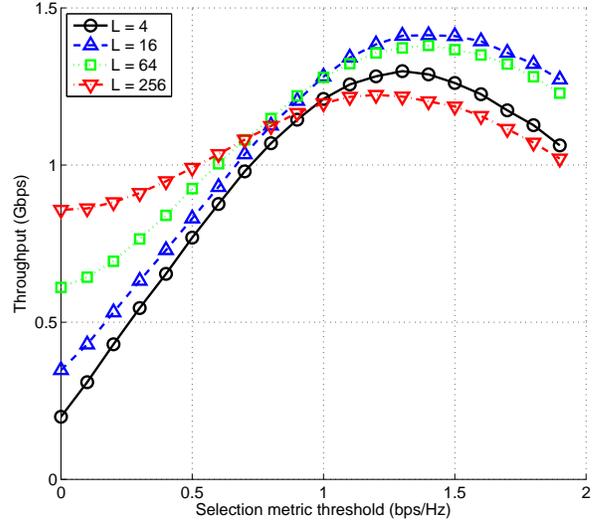}
 \caption[]{Ergodic throughput versus selection metric threshold under different numbers $L$ of beamforming pairs: $T_{\textrm{data}}=10$ $s$; mean number of active UEs per cell equals $5$.}
\label{fig:5}
\end{figure}

In Figure \ref{fig:3}, Figure \ref{fig:4}, or Figure \ref{fig:5}, it can be observed that the throughput performance becomes less sensitive when the numbers $L$ of beamforming pairs increases. This is because as the number $L$ of beamforming pairs increases, the radio channels become ``hardened'' and the relative variation of the received SNR reduces. As a result, less opportunism may be exploited by choosing an optimized selection threshold.

In Figure \ref{fig:6}, we compare the throughput performance attained by the proposed cell search and selection scheme with optimal stopping to the performance of several other cell search and selection strategies. The first scheme is the max-received-power association, where UE scans a number of cells and then selects the one that yields the maximum received power. The second scheme is to scan a \textit{fixed} number of cells and then selects the one that yields the largest selection metric. For either the max-received-power association or the max-selection-metric association,  we consider two values for the number of searched BSs in Figure \ref{fig:6}: $10$ and $30$. Figure \ref{fig:6} shows that the max-received-power association has the worst performance since it only takes into account the received power but ignores the important load factor. The max-selection-metric association takes into account both the received power and the load, but its stopping strategy of searching a fixed number of BSs is suboptimal. The throughput performance is optimized when the association is based on the proposed selection metric combined with the derived optimal stopping rule.

 \begin{figure}
 \centering
  \includegraphics[width=9cm]{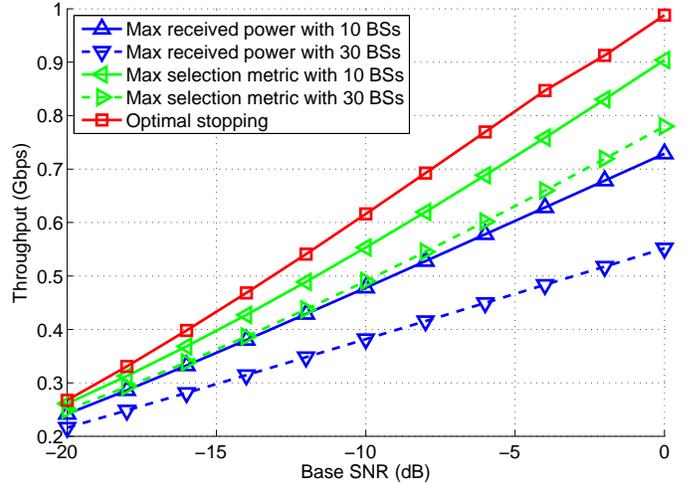}
 \caption[]{Comparison of throughput performance under different cell search and selection strategies.}
\label{fig:6}
\end{figure}

\section{Conclusions}
\label{sec:conclusion}

In this paper, we have studied initial cell search and selection in 5G networks. We propose a load-aware initial cell search and selection scheme, which is simple yet powerful. It equips the networks with a powerful access control method that facilitates load balancing. We also formulate a throughput optimization problem using the optimal stopping theory. We characterize the throughput optimal stopping strategy and the attained maximum throughput. The results show that the proposed initial cell search and selection scheme is throughput optimal with a carefully optimized connection threshold. 

This work can be extended in a number of ways. The selection metric investigated in this paper is a function of SNR. One may consider extending this metric to incorporate the effect of interference. Collisions in random access that is part of initial access procedure are not considered in this paper. It will be of interest to explore the impact of collisions on the performance.

%

\bibliographystyle{IEEEtran}
\bibliography{IEEEabrv,Reference}

\end{document}